\documentclass[aps,pra,reprint,floatfix,superscriptaddress]{revtex4-2}
\usepackage{CJK}
\usepackage[utf8]{inputenc}

\usepackage{hyperref}
\usepackage{braket,amsmath,amssymb,mathtools,amsthm}

\newtheorem{theorem}{Theorem}
\newtheorem{lemma}{Lemma}
\newtheorem{proposition}[theorem]{Proposition}

\begin{document}
\begin{CJK*}{GB}{} 
\title{A Comparison of Quantum Walk Implementations on NISQ Computers}
\author{Konstantinos Georgopoulos}
\email{k.georgopoulos2@newcastle.ac.uk}
\affiliation{School of Computing, Newcastle University, Newcastle-upon-Tyne, NE4 5TG, United Kingdom}
\author{Clive Emary}
\affiliation{Joint Quantum Centre Durham-Newcastle, School of Mathematics, Statistics and Physics, Newcastle University, Newcastle-upon-Tyne, NE1 7RU, United Kingdom}
\author{Paolo Zuliani}
\email{paolo.zuliani@newcastle.ac.uk}
\affiliation{School of Computing, Newcastle University, Newcastle-upon-Tyne, NE4 5TG, United Kingdom}
\date{\today}

\begin{abstract}
    This paper explores two circuit approaches for quantum walks: the first consists of generalised controlled inversions, whereas the second one effectively replaces them with rotation operations around the basis states. We show the theoretical foundation of the rotational implementation. The rotational approach nullifies the large amount of ancilla qubits required to carry out the computation when using the inverter implementation. Our results concentrate around the comparison of the two architectures in terms of structure, benefits and detriments, as well as the computational resources needed for each approach. We show that the inverters approach requires exponentially fewer gates than the rotations but almost half the number of qubits in the system. Finally, we execute a number of experiments using an IBM quantum computer. The experiments show the effects of noise in our circuits. Small two-qubit quantum walks evolve closer to our expectations, whereas for a larger number of steps or state space the evolution is severely affected by noise.
\end{abstract}

\maketitle
\end{CJK*}

\section{Introduction}\label{sec:intro}
This paper considers the quantum mechanical analogue of a random walk on the line, or \textit{quantum walk} \cite{Kempe-2003}. Here, the evolution of the walker is guided by a balanced quantum coin in superposition. If we imagine a quantum particle that moves freely between adjacent discrete points on a line, see Figure \ref{fig:qwalkdyn}, then at each time step, the balanced coin is flipped and the quantum state undergoes a unitary transformation, otherwise called shift. Then the particle progresses according to the state of the quantum coin, thus evolving the walk.

By evolving the state of the system in a superposition, the walker/particle can seemingly follow all possible paths, propagating quadratically further as a function of the coin flips than in the classical case \cite{Szegedy}. Significant here is also the effect of quantum interference in quantum walks, where two separate paths leading to the same point may be out of phase and cancel one another.

Quantum walks have the potential to speed up classical algorithms that are based on random walks \cite{Szegedy,Richter-2007,Shenvi-2003}. There have been many systematic studies on this subject area and many of them can lead to further in-depth analysis of more advanced quantum algorithms, such as quantum Metropolis, quantum Markov chains or quantum Monte Carlo methods \cite{Chow,Montanaro-2015,Temme,Szegedy,Dhahri-2019}. An early work from Aharonov et al. \cite{Aharonov-2001} proves that, in the context of quantum walks on graphs, the walker's propagation in the quantum case is quadratically faster than the classical random walk. The efficiency of quantum walks has been exploited in various cases in order to construct quantum algorithms \cite{Ambainis-2014,Magniez-2007} and speedup classical methods \cite{Childs-2004,Montanaro-2018}, sometimes even exponentially \cite{Childs-2002,Kempe-2005}. Quantum walks have also been realised in a number of physical systems, including photons \cite{Sansoni-2012,Rakovszky-2015,Mugel-2016,Grafe-2020,Broome-2010,Schreiber-2010}, cold atoms \cite{Karski-2009,Robens-2015} and trapped ions \cite{Schmitz-2009}.

\begin{figure}[!tb]
    \begin{center}
        \includegraphics[width=8cm]{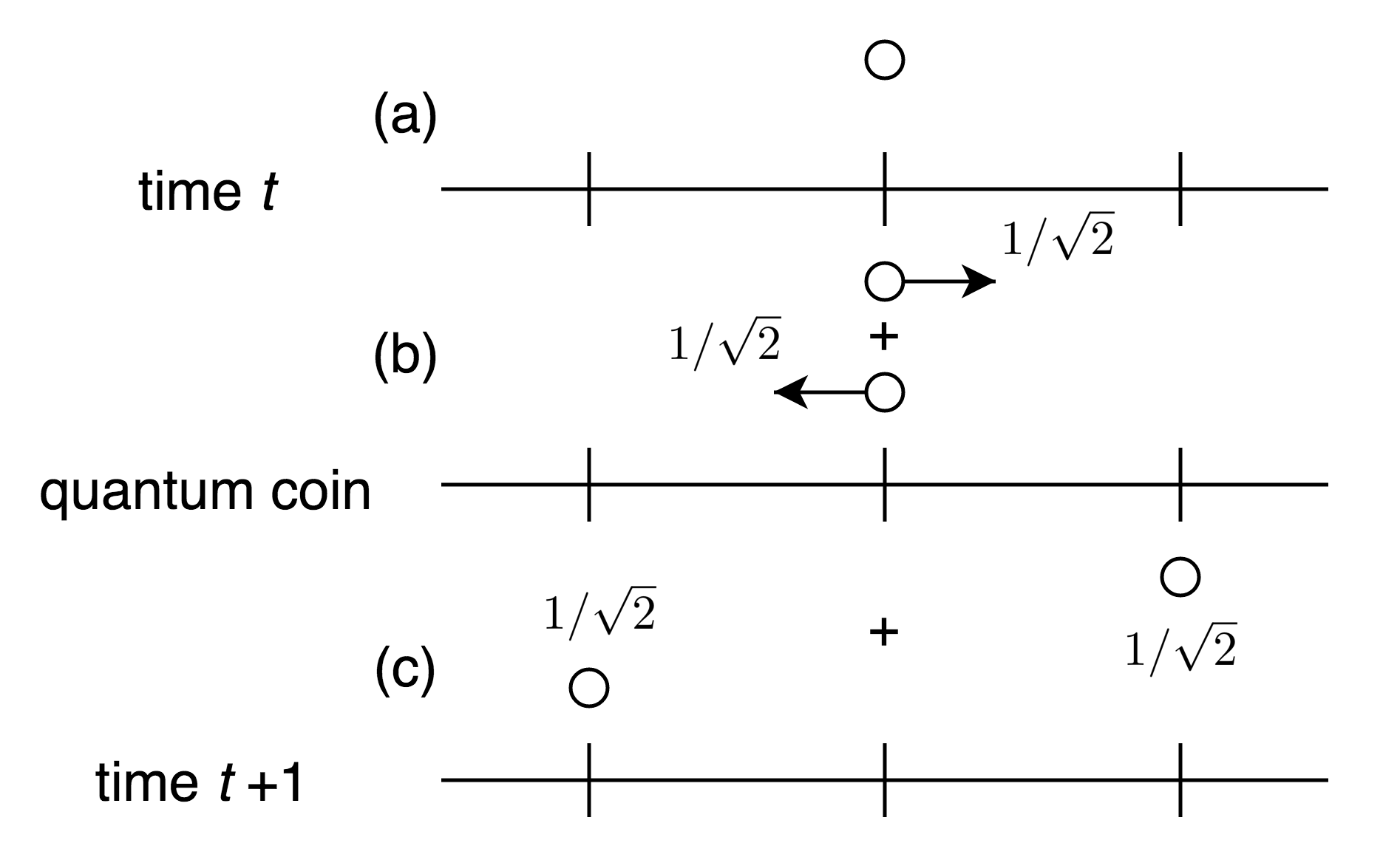}
    \end{center}
    \caption{Dynamics of the balanced quantum walk. (a) The walk begins at time $t$. (b) The flip of the quantum coin, where the particle is in equal superposition to go left or right. (c) The particle moves to generate the superposed state at time $t+1$.}
    \label{fig:qwalkdyn}
\end{figure}

The results of our work are concentrated around the comparison of two architectural approaches to quantum walks: (i) the generalised inverters approach \cite{EffWalk} and (ii) the rotational approach \cite{Gates,Liu-2008}, with no quantum walk implementation of the latter, to the best of our knowledge. Through an analysis of the computational resources necessary we find that each architecture shows opposite advantages and disadvantages. The generalised inverters approach shows smaller execution times and requires exponentially less gates whereas the rotational approach lowers the number of qubits necessary for the implementation. This result is important within the NISQ era \cite{Preskill-2018}, with the inability of existing quantum computers to reliably run even average-scale computations.

Finally, we examine the effects of noise for each implementation on a real quantum machine. Our experiments show that the results follow the theoretical expectations for small two-qubit systems, but quickly deviate for quantum walks with a three-qubit state space or larger. For the implementation we make use of IBMQ's Qiskit development kit \cite{Qiskit,IBMQExp} in order to simulate and execute the quantum circuits. Currently, there is a lot of interest in running quantum applications on NISQ systems \cite{Blank-2020,Huang_2020,Cross_2019,Martn-2019,Narenda-2017}. On the other hand, there are very few studies of quantum walks on hardware, with an early notable work by \cite{Ryan-2005} representing the first implementation of a discrete coined quantum walk implemented on a quantum-information processor. Additionally, in a recent work \cite{Acasiete-2020} the authors examine the implementation of discrete-time quantum walks on cycles, two-dimensional lattices and complete graphs, which are then executed on a different IBM quantum computer.

The paper is structured as follows. First we give a theoretical overview of the quantum walk in Section \ref{sec:qwalk}. Moving on, we discuss in Section \ref{sec:3} how a quantum walk can be implemented using generalised inverters and rotations around the basis states. Section \ref{sec:impl} offers an analysis of the implementation characteristics of each approach and a rigorous comparison of the two, before we outline the results from experimenting with both implementations in Section \ref{sec:res}. Finally, we discuss and conclude our research in Section \ref{sec:discussion}.

\section{Quantum Walks}\label{sec:qwalk}
Throughout this paper we are occupied with a particle's one-dimensional discrete-time quantum walk on a finite cycle with $N$ states and an arbitrary number of steps \cite{qWALK-11}. This quantum walk can be described as the repeated application of a unitary evolution operator, $U$. This operator acts on a Hilbert space $\mathcal{H}^{S} \otimes \mathcal{H}^{C}$, where $\mathcal{H}^{C}$ is the Hilbert space associated with a quantum coin and $\mathcal{H}^{S}$ with the state space of the walk. In order to describe the quantum walk we define the operator $U$ as
\begin{equation}
    U = S \cdot (I \otimes C) \label{walkunitary}
\end{equation}
where $S$ is the shift operator describing the walker's propagation and $C$ is the quantum coin operator. We take the quantum coin to be a Hadamard operator with well-known matrix representation
\begin{equation*}
    C = H = \frac{1}{\sqrt{2}}
    \begin{pmatrix}
        1 & 1 \\
        1 & -1
    \end{pmatrix}.
\end{equation*}

The walker, after each flip of the coin can either \textit{increase} or \textit{decrease} its position by a step of $1$. This is defined by the shift operator $S$ and can be described via increment and decrement functions, as demonstrated by \cite{EffWalk}. The mathematical description of the shift operator is
\begin{equation}
    S = S^{-} \otimes \ket{0}\bra{0} + S^{+} \otimes \ket{1}\bra{1} \label{eq:S}
\end{equation}
where $S^{+}\ket{x} \rightarrow \ket{x+1}$ moves the walker one step to the right, increasing its position, and $S^{-}\ket{x} \rightarrow \ket{x-1}$ to the left, decreasing its position.

A very important aspect of a discrete-time quantum walk's evolution is what we call the \textit{modularity property} \cite{qWALK-10}. This property expresses the relationship between the parity of the number of coin-flips, the initial position and the resulting states of the walker. For example, if the particle starts on an even position (including $\ket{0}$) then, after an odd/even number of steps the position of the particle will be a superposition of $N/2$ odd/even states.

Another characteristic of this quantum walk is its \textit{asymmetry}. After the evolution of the particle position, the probability of each state to be measured may not be the same. The reason for this phenomenon is quantum interference, which can be either constructive or destructive. This can affect the quantum walk for more than one iterations of the shift operator, $S$. Precisely, the leftwards path ($S^{-}$) interferes more destructively whereas the rightwards path undergoes more constructive interference. In other words, the asymmetry is the result of the Hadamard coin introducing bias in the path selection \cite{Krovi-2007}.

\section{Increment and Decrement Functions}\label{sec:3}
One way to implement the functions $S^{\pm}$, as expressed in Equation \eqref{eq:S}, is to use generalised inverter gates, demonstrated by \cite{EffWalk}. This implementation is discussed in Section \ref{subsec:gencnot} below. In Section \ref{subsec:rot}, we present the second approach that uses rotations around the basis states \cite{Gates,Liu-2008}.

\subsection{Using Generalised Inverter Gates}\label{subsec:gencnot}
We refer to generalised controlled operations as those controlled by more than one qubit. In this case, the operation in question will occur if and only if all the control qubits are in state $\ket{1}$. The most common example and one that is used extensively here is the three-qubit Toffoli gate, in which the target qubit will be inverted only if both the control qubits are in state $\ket{1}$.

The generalised \verb|CNOT| gates can be used to construct quantum circuits that implement the increment and decrement functions. As shown by \cite{EffWalk}, the increment and decrement functions can be realised with a single quantum circuit, but with opposite control logic. Figure \ref{fig:incredecrcirc}(a) shows the higher level circuit for one iteration of the quantum walk for a state space of arbitrary size. The general implementation of the increment and decrement circuits is shown in Figure \ref{fig:incredecrcirc}(b). 

The realisation of the individual circuits in Figure \ref{fig:incredecrcirc}(b) for a small quantum walk on a $N=8$-cycle using elementary and Toffoli gates is shown in Figure \ref{fig:inccirc}(a-b).

\begin{figure}[!ht]
    \begin{tabular}{c}
          \includegraphics[width=5cm]{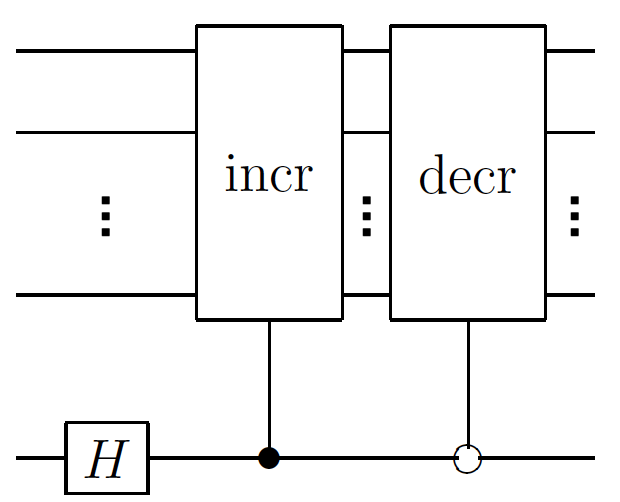} \\
          (a) \\
          \includegraphics[width=8.5cm]{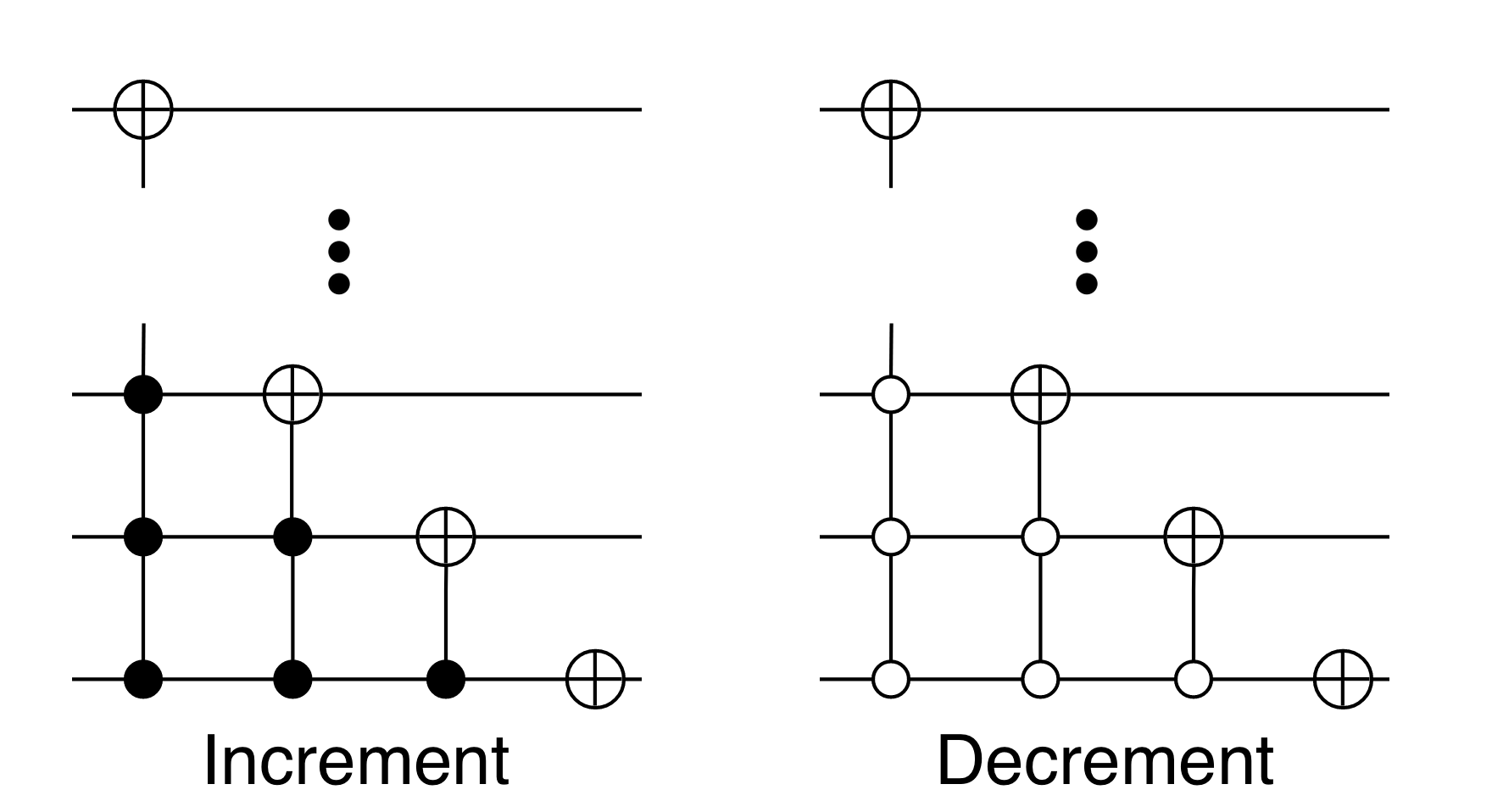} \\
          (b)
    \end{tabular}
    \caption{(a) Implementation of one step for the quantum walk of a particle. (b) Quantum circuits for increment and decrement operations. A filled control circle means that the control qubits have to be in state $\ket{1}$ in order for the operation to occur. An empty control circle means they have to be in state $\ket{0}$.}
    \label{fig:incredecrcirc}
\end{figure}



Both the increment and decrement circuits of these two figures look more complicated than their respective schematic in Figure \ref{fig:incredecrcirc}(b). The reason is the lack of direct implementation of any generalised inverters other than the Toffoli gate in IBMQ's Qiskit. Any inverter gate with more than two control qubits requires intermediate computations stored in ancilla qubits. This leads to a significant increase of the workspace (i.e. the number of qubits needed for the computation) that grows with the size of the state space. Precisely, a generalised \verb|CNOT| gate with $n_{c}$ control qubits requires additional $n_{c}-1$ ancilla qubits for the implementation (refer to Appendix \ref{ap:cnx}). For example, considering IBMQ's $15$ qubit Melbourne machine, we can implement a quantum walk on a cycle with at most $2^{8}=256$ states.



\begin{figure}[!ht]
    \begin{tabular}{c}
          \includegraphics[width=8.4cm]{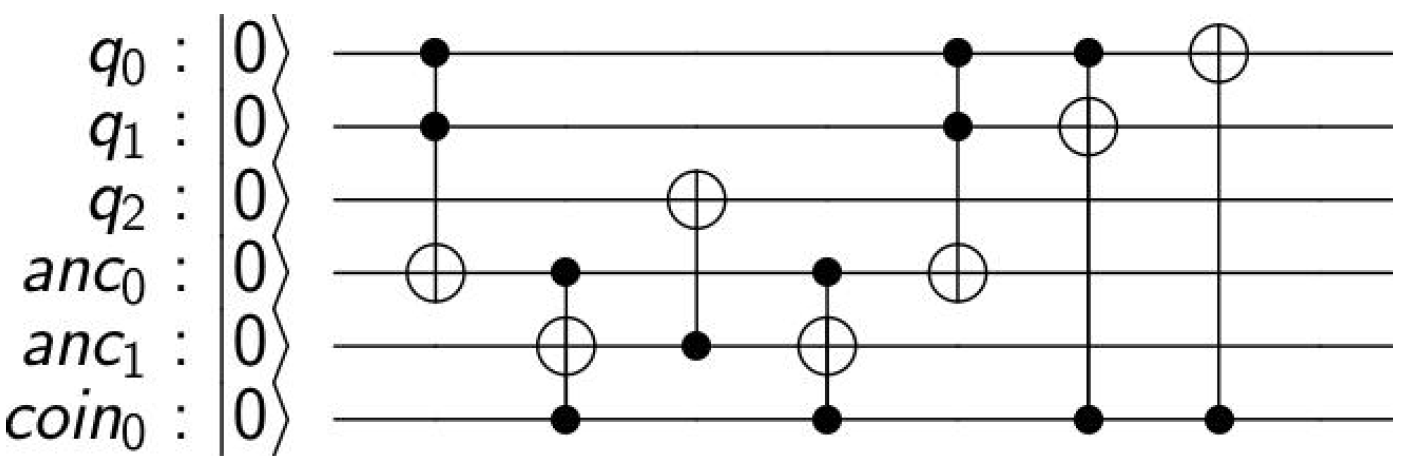} \\
          (a) \\
          \includegraphics[width=8.4cm]{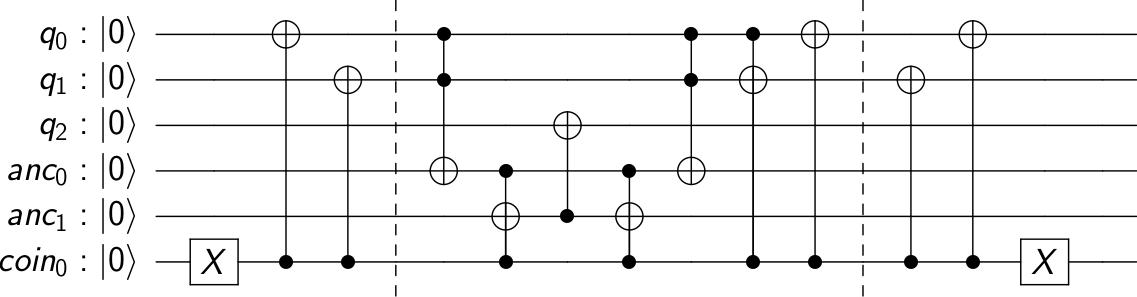} \\
          (b)
    \end{tabular}
    \caption{(a) Increment circuit. The three qubit $q_{n}$ register is the state space of the quantum walk, the coin register $coin_{0}$ represents the Hadamard coin. The ancilla qubits used for the computation are $anc_{0}$ and $anc_{1}$. (b) Decrement circuit. Important here is the need to invert all the control qubits (including the coin) at the start of the computation and uncompute them at the end.}
    \label{fig:inccirc}
\end{figure}

\subsection{Using Rotations}\label{subsec:rot}
In this section we present a solution that uses \textit{rotations} around the basis states \cite{Gates,Liu-2008} to implement the quantum walk without an ancilla register, thus lowering the vast increase in computational resources resulting from adding qubits to the system. Referring back to the previous section's example, we can use the $15$-qubit Melbourne machine to run a quantum walk with rotations on a cycle with $2^{14}=16,384$ positions, a significant rise compared to the generalised \verb|CNOT| approach.

Another benefit of the rotational approach holds when simulating quantum walks on classical machines. As the size of the state space increases exponentially with the number of qubits, classical computers very quickly start struggling to cope with the size of the workspace. The rotational implementation offers a way around this problem.

Before we show the quantum circuit for the rotational approach, we need to start by expressing two very important lemmas, first introduced and proven by Barenco et al. \cite{Gates}.

\begin{lemma}\label{lemma:rot}
    \normalfont{\textbf{\cite{Gates}}} \textit{For any unitary operator $W$ there exist operators $\Phi$, $A$, $B$ and $C$ such that $ABC=I$ and $\Phi AXBXC = W$, where $\Phi$ is a phase operator of the form $\Phi=e^{i\delta}\times I$ with $\delta \in \mathbb{R}$, $X$ is the Pauli-$X$ and $I$ the identity matrix.}
\end{lemma}

What we learn from Lemma \ref{lemma:rot} is that we can express any unitary operator and, for our case, a \verb|NOT| gate, as a sequence of operators $\Phi AXBXC$. The existence of the operator $\Phi$ compensates for the fact that the \verb|NOT| gate is not an $\text{SU}(2)$. Thus, we can narrow our efforts down to finding the appropriate $\Phi,A,B,C$ operators that suit our specific needs.

We start by presenting a quantum circuit that implements a Toffoli gate, denoted $X_{cc}$, via a number of controlled rotations, as shown in Figure \ref{fig:tofrot}. The specific rotation gates needed can be defined through the unitary matrix $R_{y}(\theta)$ defined as
\begin{equation}
    R_{y}(\theta) = 
    \begin{pmatrix}
        \cos{\theta /2} & -\sin{\theta /2} \\
        \sin{\theta /2} & \cos{\theta /2}
    \end{pmatrix} \label{eq:ry}
\end{equation}
where, in our case, we require $\theta=\pi/2$. The rotation operator can be implemented using IBM's $U_{3}$ gate, as shown in Appendix \ref{ap:u3}.

The next operation is expressed by the unitary operator $R_{z}(\phi)$, given as
\begin{equation}
    R_{z}(\phi) = 
    \begin{pmatrix}
        e^{i\phi/2} & 0 \\
        0 & e^{-i\phi/2}
    \end{pmatrix}. \label{eq:rz}
\end{equation}
Similarly to the rotation operators, we assign $\phi=\pi/2$.

\begin{figure}[!tb]
    \begin{center}
        \includegraphics[width=8.2cm]{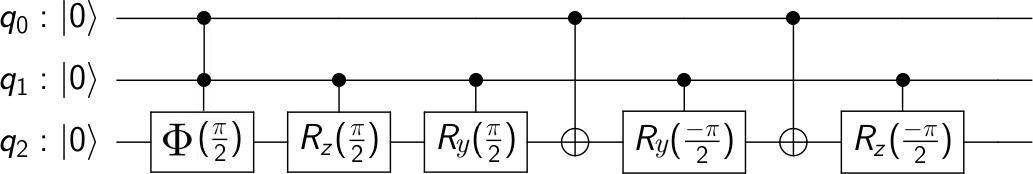}
    \end{center}
    \caption{Quantum circuit implementing a Toffoli gate ($X_{cc}$) using conditioned rotations.}
    \label{fig:tofrot}
\end{figure}

\textit{Finally}, since the inverter gate is not a special unitary (i.e. it doesn't have determinant $1$), there is the need for the additional phase gate $\Phi(\delta)$ defined as
\begin{equation*}
    \Phi(\delta) = \begin{pmatrix} e^{i\delta} & 0 \\ 0 & e^{i\delta} \end{pmatrix}
\end{equation*}
where, for our case, we identify $\delta=-\pi/2$.

Thus, we can now rewrite the rotation operation for an inverter gate, $X$, as
\begin{equation*}
    X \equiv \Phi(\pi/2)R_{z}(\pi/2)R_{y}(\pi/2)XR_{y}(-\pi/2)XR_{z}(-\pi/2)
\end{equation*}
where $A=R_{z}(\pi/2)R_{y}(\pi/2)$, $B=R_{y}(-\pi/2)$ and $C=R_{z}(-\pi/2)$. By modifying the operators to accommodate for the right matrix dimensions, we can implement the Toffoli gate.

The next step is to generalise this quantum circuit so it can accommodate more than two control qubits, i.e. create a generalised \verb|CNOT| using rotations. In order to do this we can use another lemma from Barenco et al. \cite{EffWalk}. In this context, we need to introduce the notation $\wedge_{n-1}(U)$ as used by \cite{Gates,Liu-2008}. For any unitary matrix $U=\left( \begin{array}{ll}{u_{00}} & {u_{01}} \\ {u_{10}} & {u_{11}}\end{array}\right)$ and $m \in\{0,1,2,\dots\}$ we define the $(m+1)$-bit $(2^{m+1}$-dimensional) operator $\wedge_{m}(U)$ as
\begin{equation*}
    \wedge_{m}(U)(\ket{x_{1},\dots,x_{m},y})=
\end{equation*}
\begin{equation*}
    \begin{aligned}
        \left\{\begin{array}{ll}{u_{y 0}\left|x_{1}, \ldots x_{m}, 0\right\rangle + u_{y 1}\left|x_{1}, \ldots, x_{n}, 1\right\rangle} & {\text {if} \wedge_{k=1}^{m} x_{k}=1} \\ {\left|x_{1}, \ldots, x_{m}, y\right\rangle} & {\text {if} \wedge_{k=1}^{m} x_{k}=0}\end{array}\right.
    \end{aligned}
\end{equation*}
where $\wedge_{k}$ denotes the \verb|AND| operation of the relevant $k$ values.

\begin{lemma}\label{lemma:genrot}
    \normalfont{\textbf{\cite{Gates}}} \textit{For any unitary $W$, a $\wedge_{n-1}(W)$ gate can be simulated by a network of rotation and phase operators, as shown in Figure \ref{fig:gentofrot}, with $\Phi$, $A$, $B$ and $C$ as in Lemma \ref{lemma:rot}.}
\end{lemma}

\begin{figure}[!ht]
    \begin{center}
        \includegraphics[width=8cm]{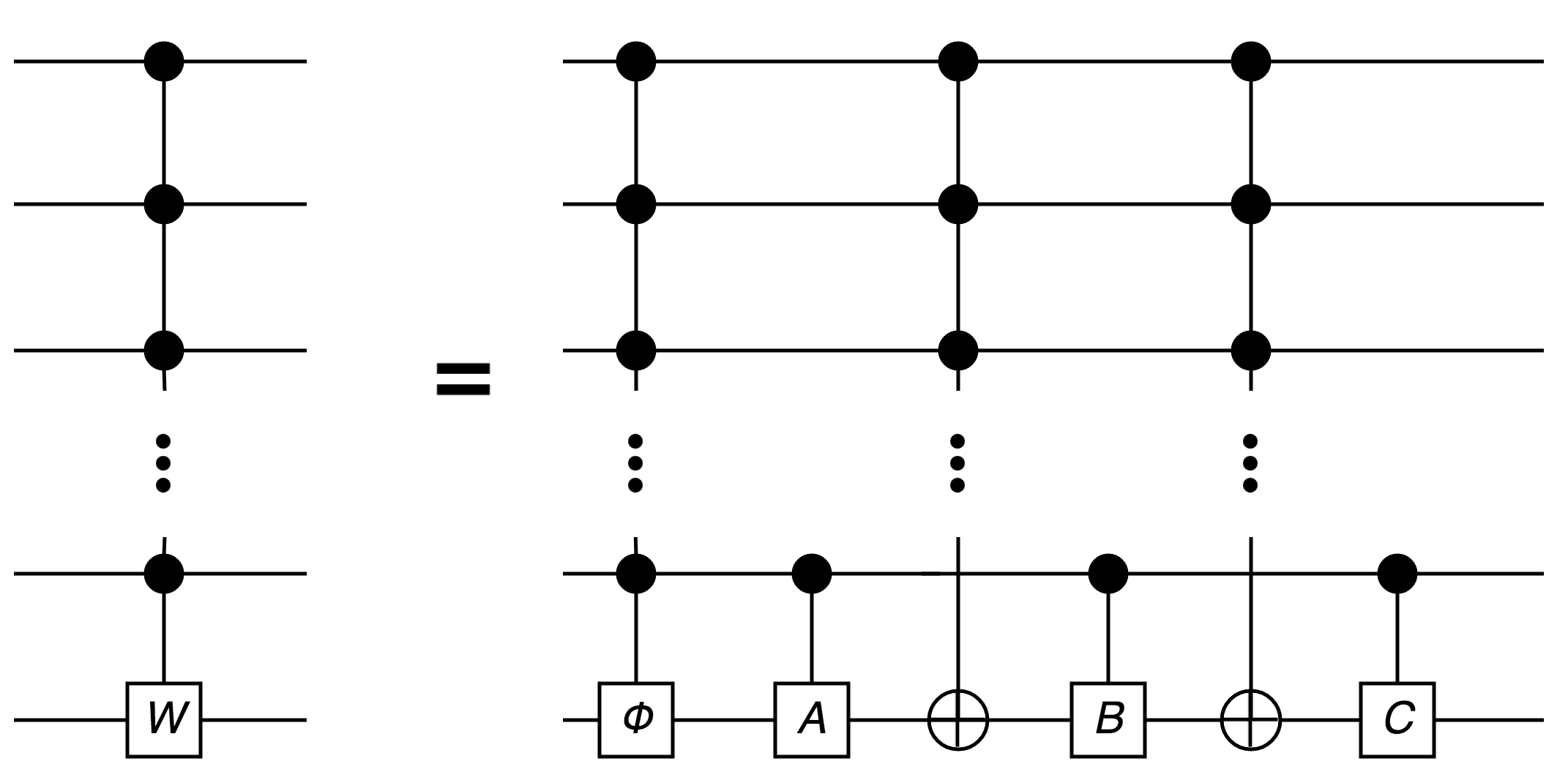}
    \end{center}
    \caption{Generalised rotational network that implements a unitary controlled by an arbitrary number of control qubits.}
    \label{fig:gentofrot}
\end{figure}

Lemma \ref{lemma:genrot} describes a way to expand any generalised unitary with an arbitrary number $m$ of control qubits to a network of controlled rotations and generalised \verb|CNOT|s of the form $\Phi AXBXC$. It is easy to see that, if $W=X$, we can iteratively expand each one of the generalised \verb|CNOT| gates to such a network. The expansion will stop when the generalised inverter gates end up being regular Toffoli gates. After the transformation of the initial approach to rotation operations, the $2\times 2$ operator $W$ applied to the target qubit is the regular inverter
\begin{equation}
    W = X =
    \begin{pmatrix}
        0 & 1 \\
        1 & 0
    \end{pmatrix}, \label{eq:W}
\end{equation}
with the dimensions of the matrix representation defined according to the dimensionality of the workspace.

Thus, we can produce a quantum circuit that implements a generalised \verb|CNOT| gate with an arbitrary number of control qubits without depending on the use of any ancilla qubits. This logic can be applied to any unitary operator \cite{Gates,Liu-2008}. 

\begin{figure}[!t]
    \centering
    \includegraphics[width=8.6cm]{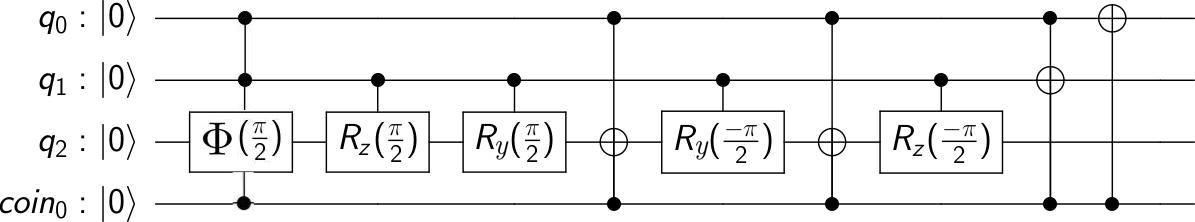}
    \caption{Rotational implementation of an increment circuit for a three qubit state space and one qubit coin.}
    \label{fig:incrot}
\end{figure}

We can now integrate this implementation to the increment and decrement circuits. We substitute the generalised \verb|CNOT| gates with the network described in Lemma \ref{lemma:genrot}. Any Toffoli, \verb|CNOT| or inverter \verb|X| gate remains the same. A visualisation of an increment quantum circuit on $4$ qubits is shown in Figure \ref{fig:incrot}. The decrement circuit will follow similar logic with the difference that all the control qubits have to be inverted at the start of the computation.

\section{Implementation Characteristics}\label{sec:impl}
For our experiments we implement a quantum walk on an $N$-cycle. This means that the walker can be in $N$ different states and needs $n=\log N$ qubits to represent it. One additional single-qubit quantum register is needed for the quantum coin and, in the case of the generalised inverters approach, an ancilla register. We find that the number of qubits needed for the inverters approach increases as $2\log N$, whereas for the rotational approach as $1+\log N$.

The difference in the size of the workspace between the two approaches increases linearly with the state space of the walk, i.e. as we move onto larger cycles. On the other hand, the rotational circuit needs an increasingly larger number of gates than the generalised \verb|CNOT| approach. This is shown with the following propositions.

\begin{proposition}\label{prop:gatescnot}
    The number of gates that participate in the generalised inverter implementation of the quantum walk circuit increases polylogarithmically with the size of the state space, $N$, as $\mathcal{O}(\log^{2}{N})$.
\end{proposition}

\begin{proof}
    For a state space $N\geq 8$, any gate that needs more than two control qubits is expanded to a network with ancilla qubits. The number of gates necessary for this expansion can be expressed as $2\sum_{n_{c}=3}^{\log{N}}(2n_{c}-1)$, where $n_{c}$ is the number of control qubits necessary for each operation. The additional gates needed will be the inverters with two or less control qubits and the Hadamard gate. 
    For a state space of $N<8$ there will be no operations with more than two control qubits and the number of gates will be simply calculated by the inverters and the Hadamard gate.
    
    Thus, the number of gates for the generalised inverter implementation can be expressed as
    \begin{equation}
        \nu_{\operatorname{c}} = \begin{cases}
            2\sum_{n_{c}=3}^{\log{N}}(2n_{c}-1) + 2\log{N} + 5, \; \text{if} \; N\geq 8 \\ \\
            2\log{N} + 5, \; \text{if} \; 2 \leq N < 8
        \end{cases} \label{eq:gatescnot}
    \end{equation}
    From the above equation we can see that the sum $\sum_{n_{c}=3}^{\log{N}}(2n_{c}-1)$ provides the dominant growth rate. We find that
    \begin{equation*}
        \begin{aligned}
            \sum_{n_{c}=3}^{\log{N}}(2n_{c}-1) &= \sum_{n_{c}=1}^{\log{N}}(2n_{c}-1) - \sum_{n_{c}=1}^{2}(2n_{c}-1) \\
            &= \log^{2}{N} - 4,
        \end{aligned}
    \end{equation*}
    as the sum $\sum_{n_{c}=1}^{\log{N}}(2n_{c}-1) = \log^{2}{N}$ is the known sum of the first $\log{N}$ odd natural numbers and $\sum_{n_{c}=1}^{2}(2n_{c}-1)=4$. Thus, the number of gates increases with the size of the state space, $N$, as $\mathcal{O}(\log^{2}{N})$.
\end{proof}

\begin{proposition}\label{prop:gatesrot}
    The number of gates that participate in the rotational implementation of the quantum walk circuit increases linearly with the size of the state space, $N$, as $\mathcal{O}(N)$.
\end{proposition}

\begin{proof}
    For $N\geq 8$ any operation with more than three control qubits will be expanded according to the network of Lemma \ref{lemma:genrot}. It is evident that every inverter as in Figure \ref{fig:gentofrot} will need to be expanded to a rotational network until all gates need only two control qubits. This leads to a number of $\sum_{j=3}^{\log{N}}\left[10\sum_{n_{c}=j}^{\log{N}}(2^{n_{c}-j})\right]$ gates before the last step as well as $\sum_{j=3}^{\log{N}}(2^{2-j+\log{N}})$ gates on the last step of the expansion. For operations with two or less control qubits, there is no need for rotations or expansions.
    
    For $N<8$ there will be no operations with more than two control qubits and the circuit will not need a rotational approach.
    
    The number of gates for the rotational implementation, $\nu_{r}$, can be expressed as
    \begin{equation}
        \nu_{r} = \begin{cases}
            \sum_{j=3}^{\log{N}}\left[(2^{2-j+\log{N}})+10\sum_{n_{c}=j}^{\log{N}}(2^{n_{c}-j})\right] \\ + 2\log{N} + 5, \text{if} \; N \geq 8 \\ \\
            2\log{N} + 5, \; \text{if} \; 2 \leq N < 8
        \end{cases} \label{eq:gatesrot}
    \end{equation}
    From equation \eqref{eq:gatesrot} we find that again the sum provides the dominant growth. In this case it represents the well known sum of a geometric progression, where the largest growth would be given by $2^{\log{N}}=N$. Thus, the number of gates increases linearly with the size of the state space, $N$, as $\mathcal{O}(N)$.
\end{proof}

Liu et al. \cite{Liu-2008} present an anaysis on how $n$-qubit controlled unitaries can be realised by one-qubit and \verb|CNOT| gates using exponential and polynomial complexity respectively. Our propositions, however, analyse the complexity of architectures implementing a quantum walk, not simply multi-qubit controlled gates.

As evident from Propositions \ref{prop:gatescnot} and \ref{prop:gatesrot}, the number of gates needed by the rotational implementation increases exponentially faster than for the generalised inverters. Of course, since the exact same circuit is repeated after every coin-flip, for arbitrary coin-flips, $t$, and the same state space $N$ the number of gates is $t$ times the number of gates needed for each implementation.

By studying the propositions we can draw some conclusions regarding the complexity of each approach. For a small state space, the generalised inverters implementation is more efficient in terms of gates, with the efficiency increasing with $N$. On the other hand, the inverters approach is less efficient in terms of workspace size, quickly surpassing the number of qubits that classical computers can simulate, or the capacity of near-term quantum machines.

\begin{figure*}[!t]
    \begin{tabular}{cc}
          \includegraphics[width=6.5cm]{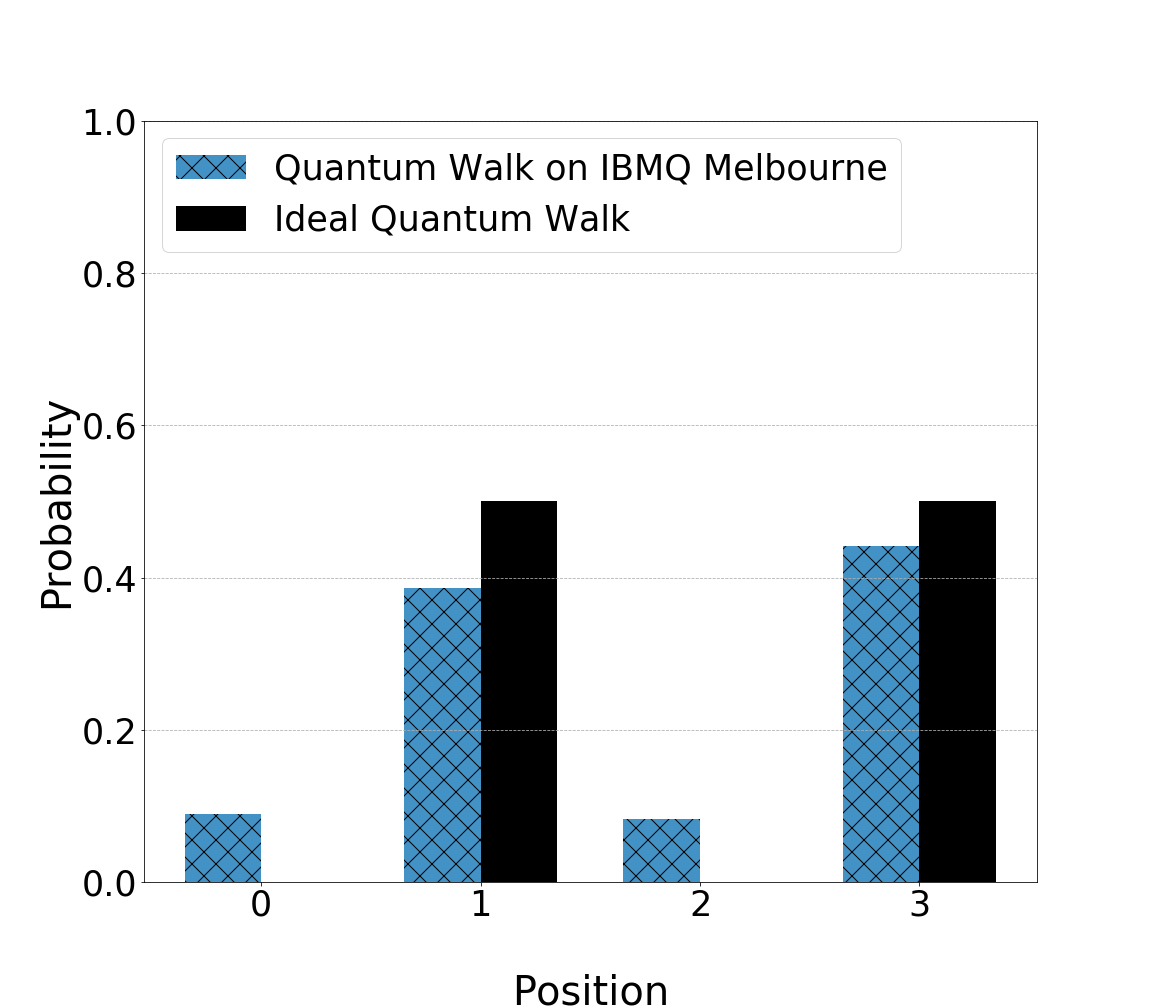} & \hspace{3em} \includegraphics[width=6.5cm]{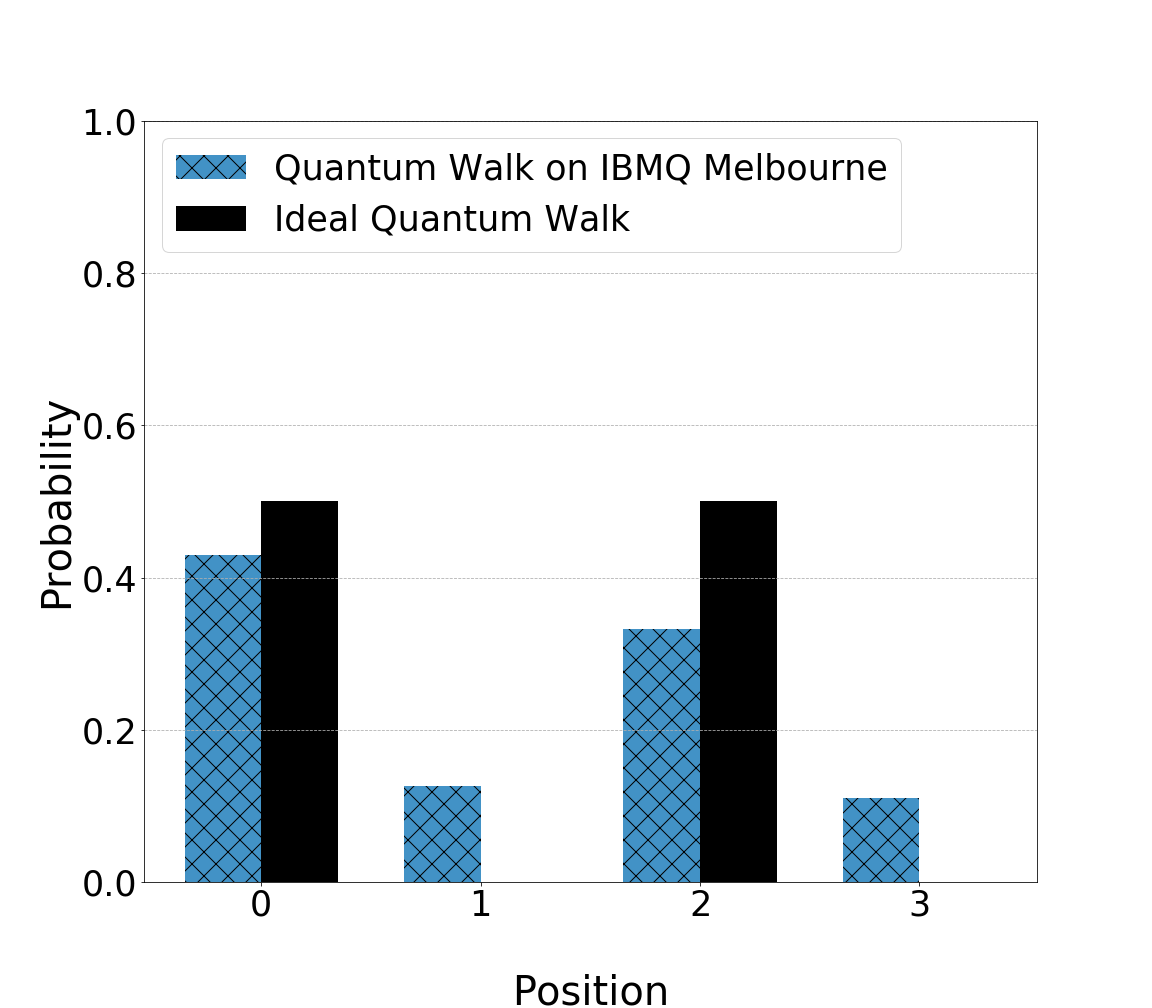} \\
        (a) & \hspace{3em} (b) \\[6pt]
          \includegraphics[width=6.5cm]{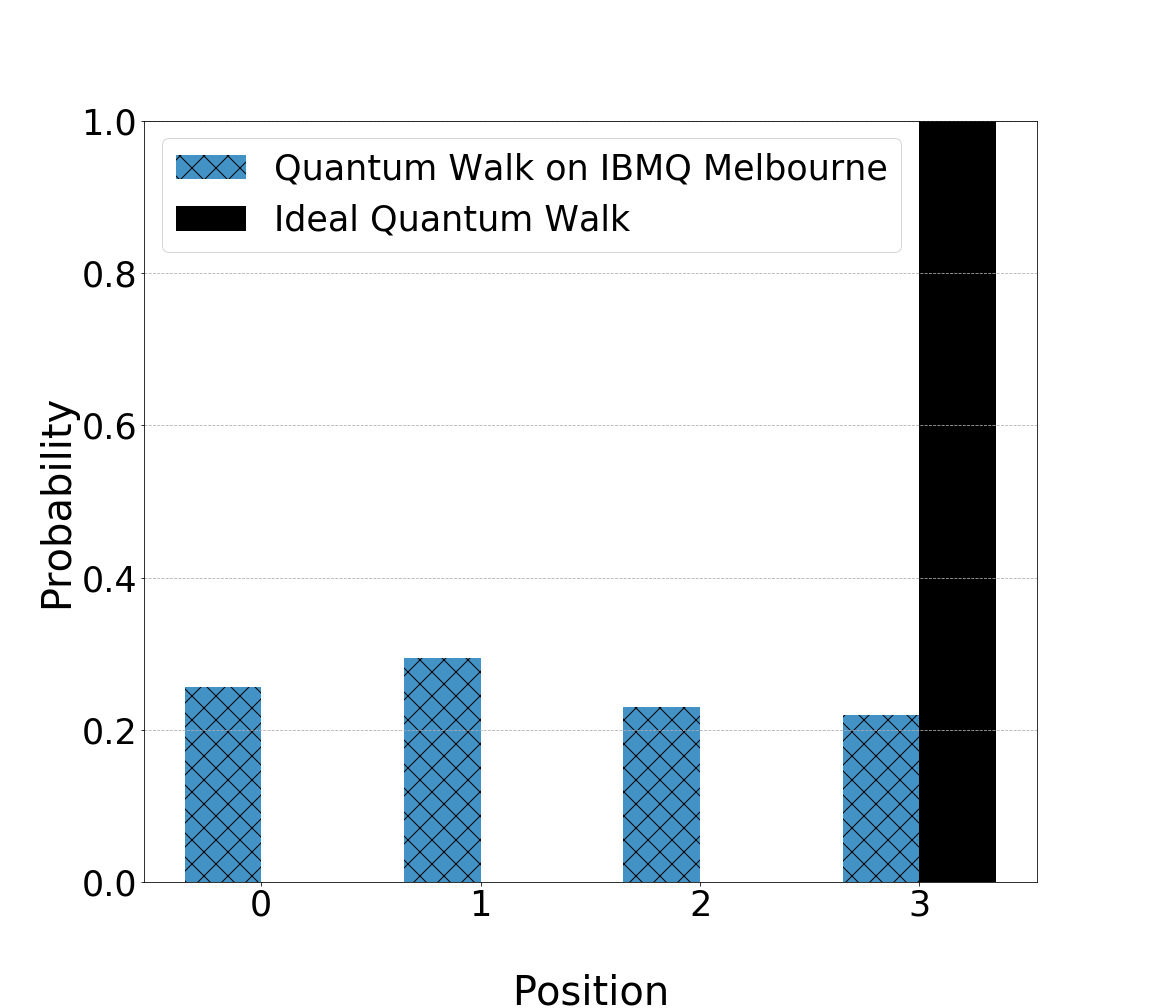} & \hspace{3em}
          \includegraphics[width=6.5cm]{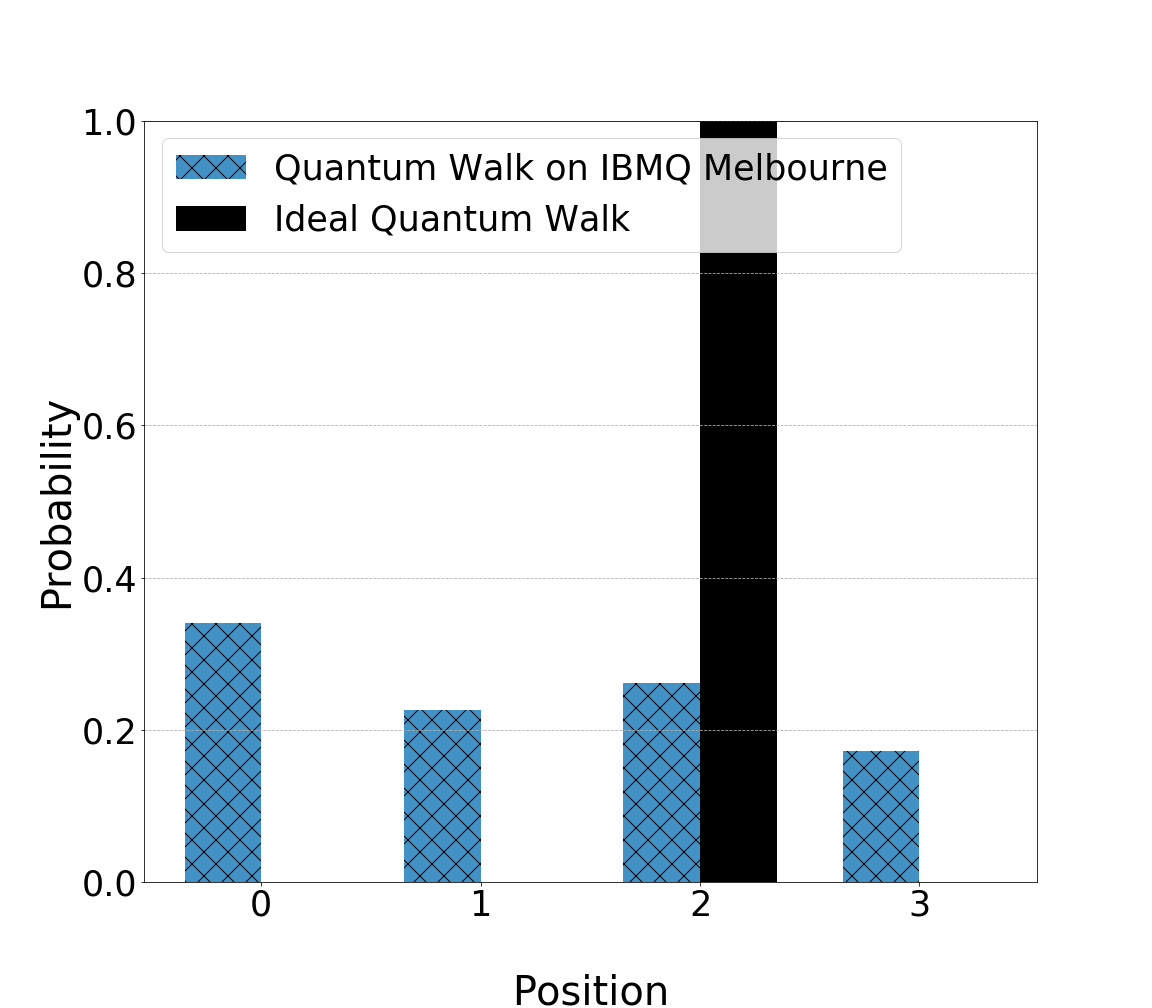} \\
        (c) & \hspace{3em} (d) \\[6pt]
    \end{tabular}
    \caption{(Color online) Probability distributions of two-qubit quantum walks on the IBMQ Melbourne computer (crossed bar) and on an ideal simulator (solid bar) for (a) one step, (b) two steps, (c) three steps and (d) four steps, with generalised inverters approach (rotations are not needed); $95\%$ confidence intervals are smaller than $10^{-3}$, hence are not displayed.}
    \label{fig:2qQW}
\end{figure*}

\section{Experimental Results}\label{sec:res}
In this section we present the results from running quantum walk experiments on $4$- and $8$-cycles as a simulation on a classical machine and on the quantum computer. The state space is of size two and three qubits respectively. We initialise the walker on position $\ket{0}$.

The results of the simulations and experiments for both approaches on two and three qubit states are given in Figures \ref{fig:2qQW} and \ref{fig:3qQW} respectively. We first discuss the results of the classical simulations before moving to the quantum computer. The code is available at \url{https://github.com/pzuliani/qwalk}.

\subsection{Experiments on Noise-free Simulator}\label{subsec:simres}
The simulations were run on a MacBook Pro $2017$ computer with a $2.3$ GHz Intel Core i5 processor and $16$ GB of memory. There are three important points to observe here: (i) the asymmetry of the probability distributions, (ii) the modular behaviour and (iii) the variance of the quantum walk. The effects of the asymmetry can be seen as the imbalance in the ideal probability distribution, as for example in Figure \ref{fig:3qQW}(c) with state $\ket{7}$ appearing with higher probability than the rest.

\begin{figure}[!ht]
    \begin{tabular}{c}
          \includegraphics[width=6.85cm]{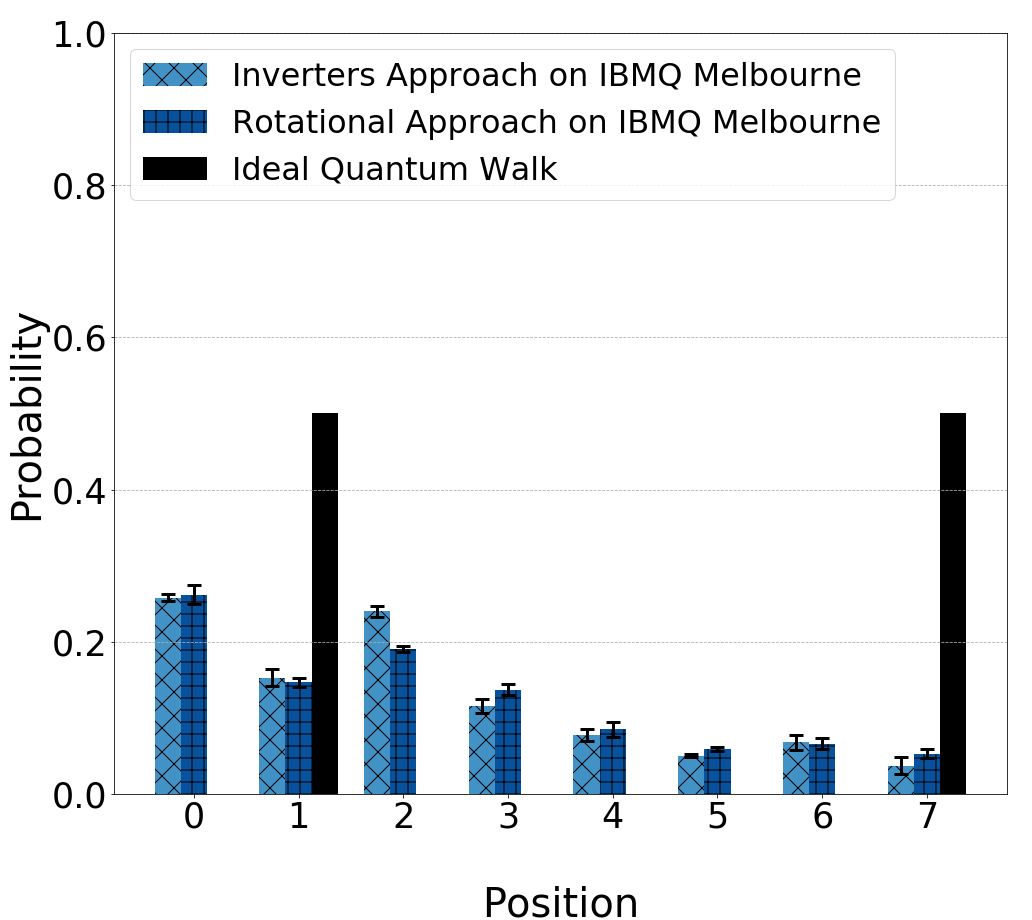} \\
          (a) \\
          \includegraphics[width=6.85cm]{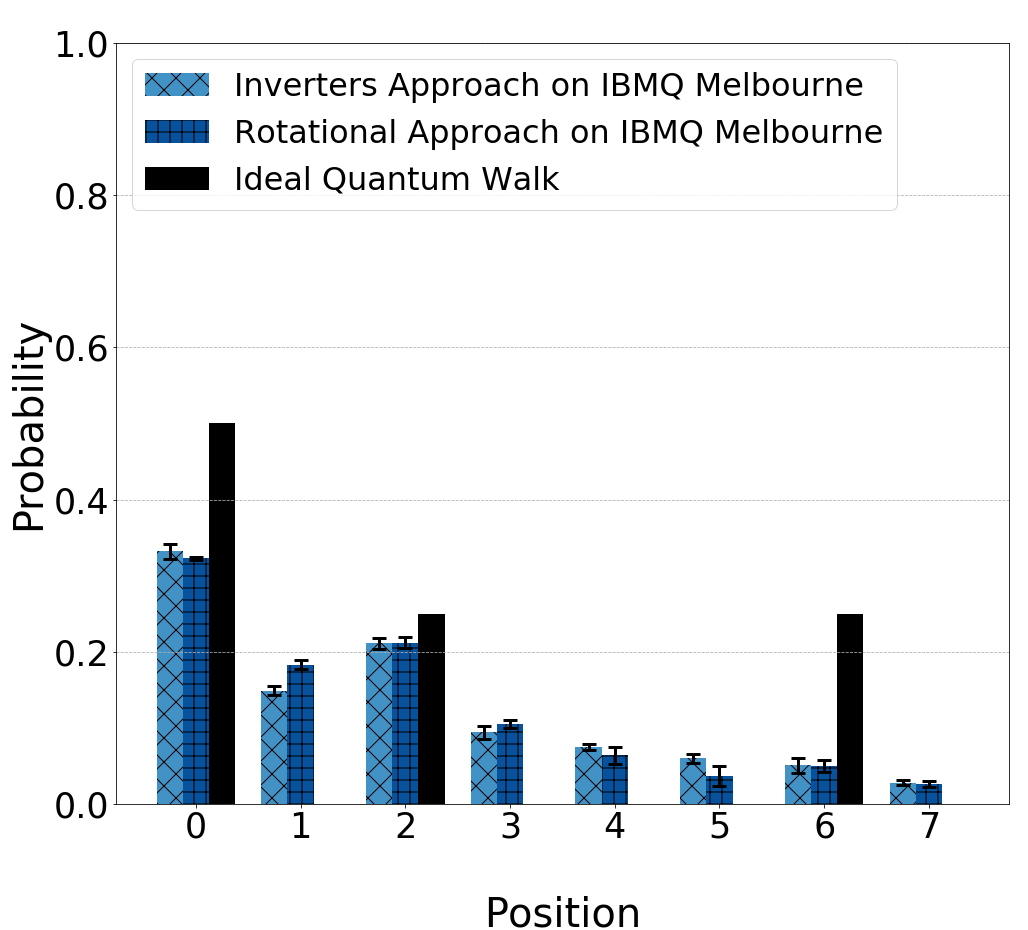} \\
          (b) \\
          \includegraphics[width=6.85cm]{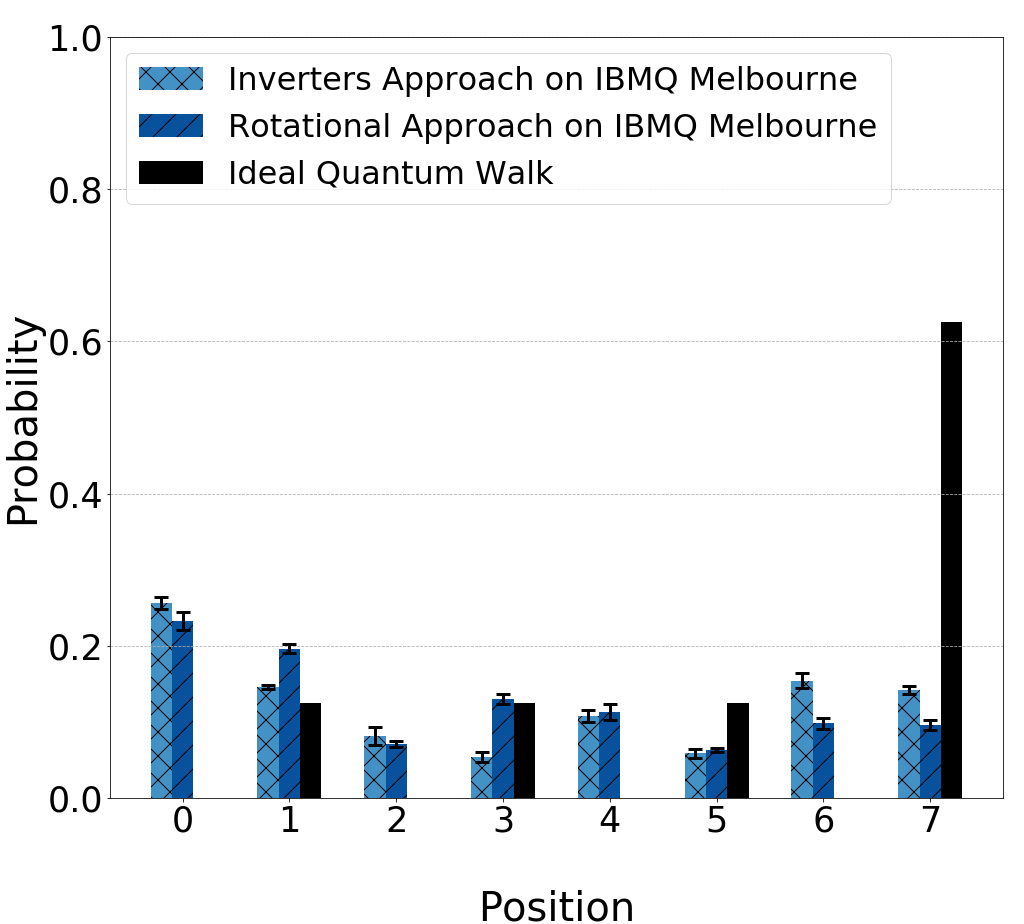} \\
          (c)
    \end{tabular}
    \caption{(Color online) Probability distributions of quantum walks on three qubits for (a) one step, (b) two steps and (c) three steps using IBMQ's Melbourne machine for the generalized inverter (crossed bar) and rotational (tiled bar) approach and an ideal simulator (solid bar), on which both approaches are identical. Error bars are calculated with $95\%$ confidence intervals.}
    \label{fig:3qQW}
\end{figure}

The second point involves the modularity of the quantum walk. Having initialised the walker on state $\ket{0}$ (considered an even state) and evolving for an odd or even number of steps, we predict that the measurement outcome will be an odd or even one respectively. Indeed that is the case, with the outcomes measured also satisfying that only $N/2$ states are observed. This stands true for both approaches.

The final point refers to the variance of the quantum walk. It has been proven that Markov chains show near-quadratic increase in the variance with respect to time as opposed to their classical analogues \cite{Szegedy}. Theoretically, the variance, $\sigma^{2}$, as a function of the coin flips, $t$, can be calculated as \cite{Orthey-2017}
\begin{equation}
    \sigma^{2} = \frac{\sqrt{2} -1}{2} t^2  \approx 0.2 \times t^2 \label{eq:var}
\end{equation}
By computing the simulated quantum walk variance we can verify this quadratic tendency for both implementations, as depicted in Figure \ref{fig:var}, concluding that both our quantum circuits are likely to be correct implementations of a quantum walk.

\begin{figure}[!tb]
    \begin{center}
        \includegraphics[width=8cm]{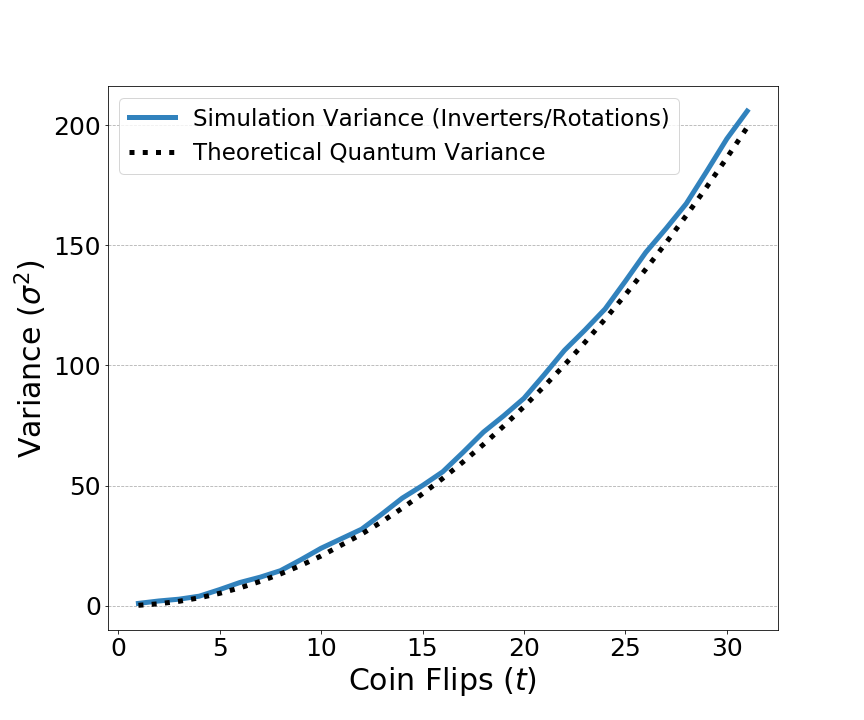}
    \end{center}
    \caption{(Color online) Variance of the quantum walk as a function of the coin flips. Here the simulation variance (solid line) is the one computed on a $N=256$-cycle via the simulations of either the inverter or the rotational circuit. Theoretical variance (dotted line) is calculated from equation \eqref{eq:var}.}
    \label{fig:var}
\end{figure}

Finally, the simulation runtimes for different number of qubits on an $8$-cycle for the two approaches are presented in Figure \ref{fig:runtimes}. It is visible that the runtime of the quantum walk circuit increases exponentially with the number of qubits. For the generalised inverters approach, our classical machine is unable to simulate the circuit for $n>16$ due to lack of memory. On the other hand, the rotational approach is able to simulate the quantum walk for state spaces larger than $16$ qubits.

\subsection{Experiments on Quantum Computer}\label{subsec:qcres}
The results we obtain by executing the experiments on the quantum processor are quite different. Both the quantum walk experiments are executed on IBMQ's $15$-qubit Melbourne machine. Due to limitations on the number of iterations we can submit to the machine, we repeat the quantum walk $1,000$ times in what constitutes a batch of trials. We repeat such batches of trials $100$ times and muster the probability distributions from the average results of the $100,000$ repetitions of the experiment. The resulting probability distribution for three steps of the quantum walk on a $4$- or $8$-cycle can be seen in Figure \ref{fig:2qQW} and Figure \ref{fig:3qQW} respectively. It is noteworthy here that since for a quantum walk on a $4$-cycle there are no inverters with more than two control qubits, a rotational implementation is not needed.

We can see from Figure \ref{fig:2qQW} and Figure \ref{fig:3qQW} that the empirical distributions differ greatly from the simulations. None of the properties that we expect to see from a quantum walk are present. More specifically, the process no longer exhibits modular behaviour as we see states that should not occur, especially on the walks with larger state space or number of steps. Unfortunately, the resulting probability distributions are completely different to the theoretical ones, making it difficult to correctly calculate the variance. Thus, no remarks regarding the variance can be made with certainty.

\begin{figure}[!ht]
    \begin{center}
        \includegraphics[width=8cm]{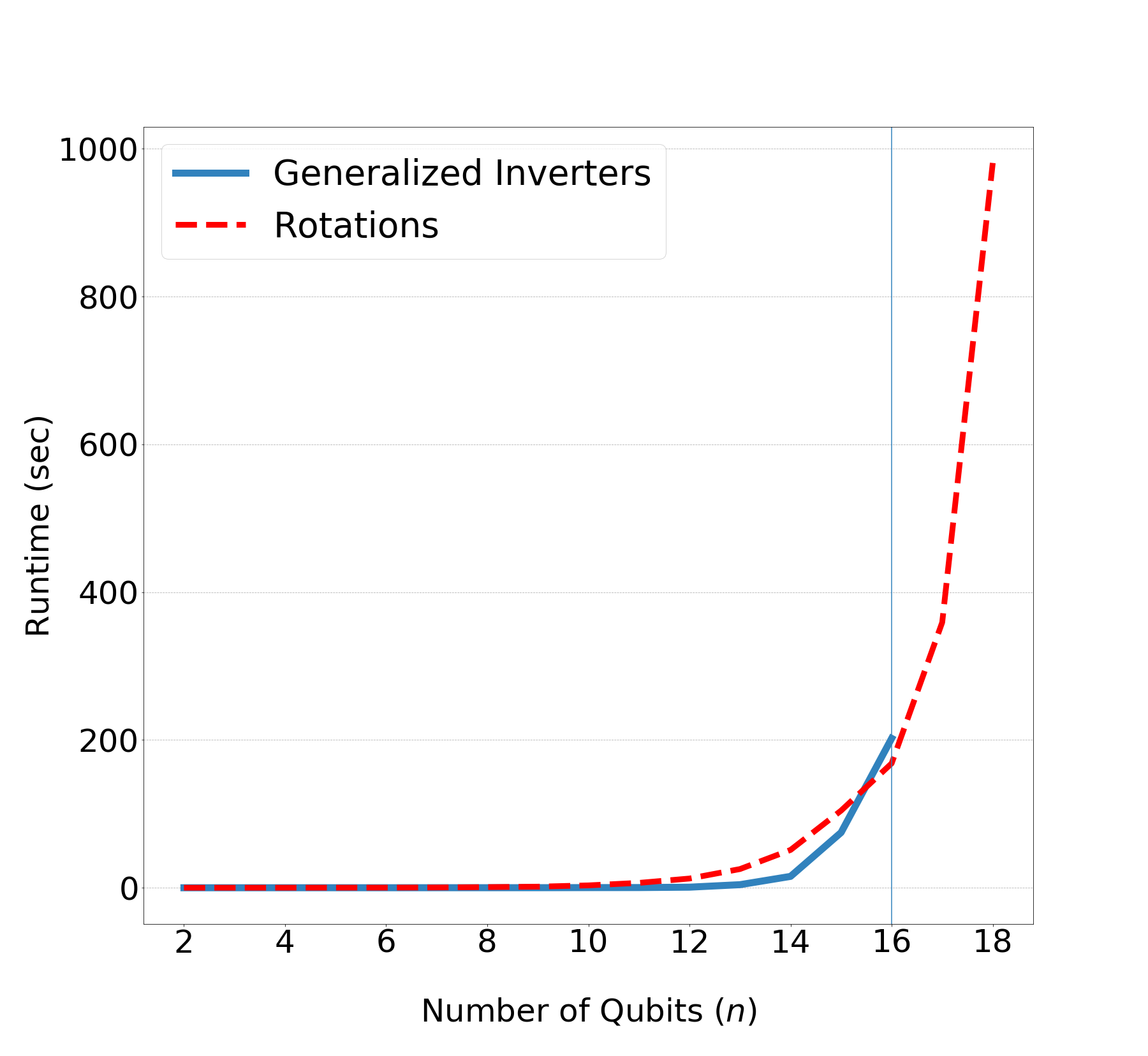}
    \end{center}
    \caption{(Color online) The simulation runtimes for the generalized inverters (solid line) and rotational (dashed line) implementations. For $n>16$ our classical machine is unable to simulate the generalised inverter implementation. Simulations are run for one step. The number of qubits $n$ only refers to the state space of the walk.}
    \label{fig:runtimes}
\end{figure}

In another paper, Balu et al. \cite{Balu-2018} realise a topological quantum walk on a $4$-cycle for one step of the walker. This implementation differs from both approaches we are concerned with in our work. The results capture the topological effects of the quantum walk at a single step, but fail to do so for more steps due to noise. Qualitatively, the results of \cite{Balu-2018} match the theoretical expectations for one step of the walk in a way similar to ours.

Taking into account the nominal execution time of the different gates that participate in the circuits, as given by IBMQ, we can calculate the overall runtime of the quantum computations. For one step of the walk with three-qubits state space, the generalised \verb|CNOT| circuit's execution time is approximately $42$ microseconds and for the rotational circuit $76$ microseconds. Similarly, for one and two steps of the walk on two qubits, the execution times are $10.5$ and $21$ microseconds respectively. These execution times are purely for the operations themselves, meaning we do not take into account the state preparation (if it occurs), the measurement or the buffer time between operations on the same qubit.

In general we find that shorter quantum walk circuits, for example one or two steps of a two-qubit quantum walk, can generally provide results closer to our expectations, as shown in Figure \ref{fig:2qQW}(a-b). This is due to the low execution time of the quantum circuit, as calculated above, compared to the average coherence time of the three qubits participating in the circuit (i.e. $64.25$ microseconds) and the smallest decoherence time among these qubits (i.e. $56$ microseconds). For more than two steps we find that the distribution starts to deviate from the expected (Figure \ref{fig:2qQW}(c-d)).

As shown in Figure \ref{fig:3qQW}(a-c), the effects of the noise remain intense for both approaches on an $8$-cycle. The average coherence of the qubits participating in the circuit is found to be $58.6$ microseconds for the six qubits in the inverters approach and $63.6$ microseconds for the four qubits of the rotational approach.

\begin{figure}[!tb]
    \begin{center}
        \includegraphics[width=8.6cm]{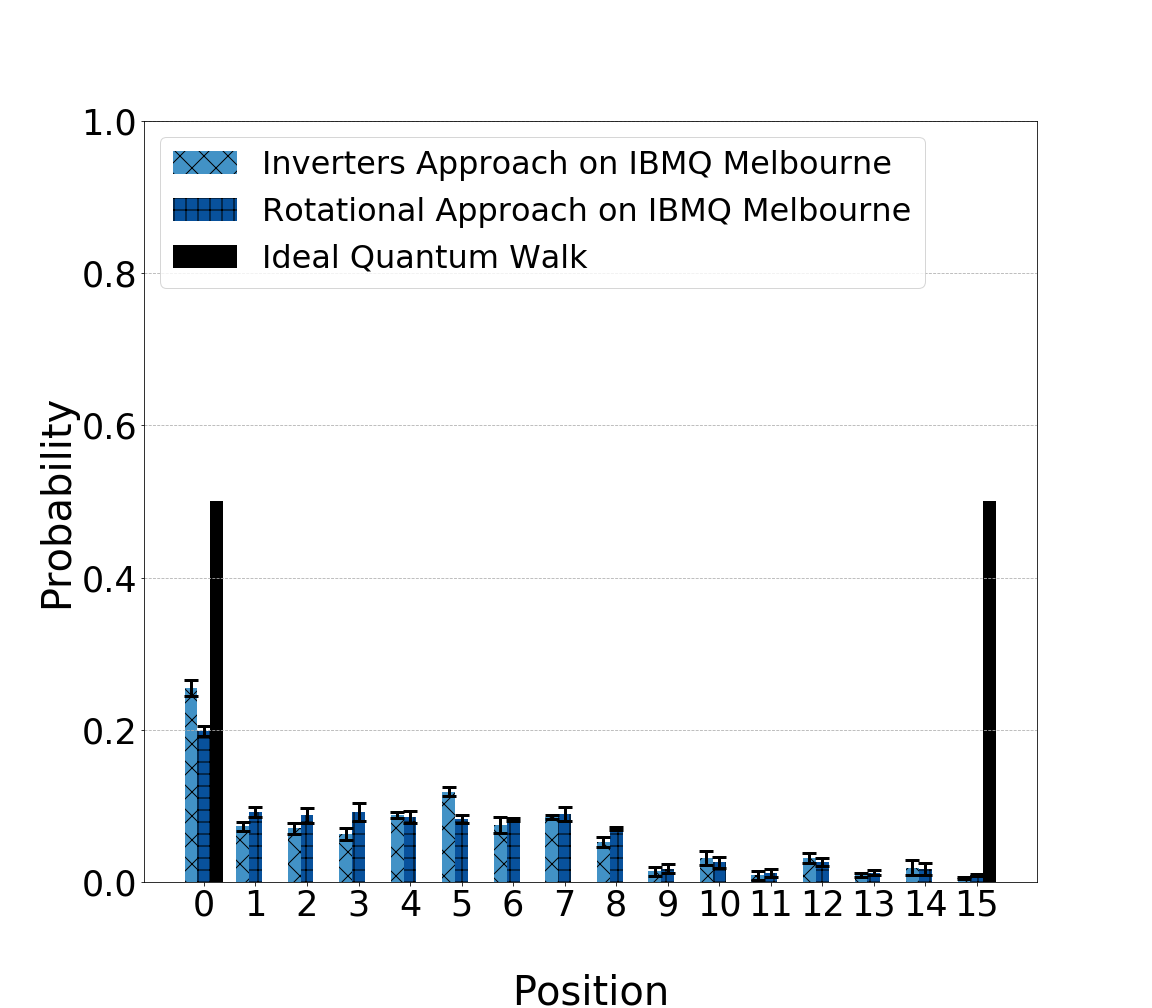}
    \end{center}
    \caption{(Color online) Probability distributions of generalized inverter (crossed bar), rotational (tiled bar) and ideal (solid bar) quantum walks on four qubits for one step. Error bars are calculated with $95\%$ confidence intervals.}
    \label{fig:4qQW}
\end{figure}

Finally, similar results are obtained on larger state spaces. One example is one step of the walk on a four qubit state space, shown in Figure \ref{fig:4qQW}. Thus, we can conclude that, for one or two steps on a two qubits state space, the quantum walk behaves relatively close to the expectation, whereas for more than three steps or for a three-qubit state space, where the runtime of the circuit is longer and approaches the coherence time, the effects of the noise are overwhelming.

\section{Discussion}\label{sec:discussion}
Throughout this paper, two approaches are used to implement a quantum walk: (i) the generalised inverters \cite{EffWalk}, and (ii) the rotational approach. Our experiments show that using generalised inverters keeps the implementation simple and the circuit shallow but requires an ancilla register. The number of ancilla qubits increases linearly with the number of control qubits quickly leading to a large workspace, limiting the capabilities of our experiment. The rotational approach deals with this limitation by rendering the ancilla register obsolete, allowing us to experiment with a much larger state space for our quantum walks. The disadvantage of the rotational approach is the larger complexity of the resulting circuit.

It is evident from our experiments that the two implementations of quantum walks offer opposite advantages and disadvantages. Implementing a quantum walk with generalized inverters shows smaller execution time and requires exponentially less gates as a function of the size of the state space, or in other words, smaller circuit depth. On the other hand, the rotational circuit requires less qubits to be used in the workspace, as there is no need for an ancilla register, but the circuit is much deeper than the first approach.

\begin{figure}[!tb]
    \begin{center}
        \includegraphics[width=8cm]{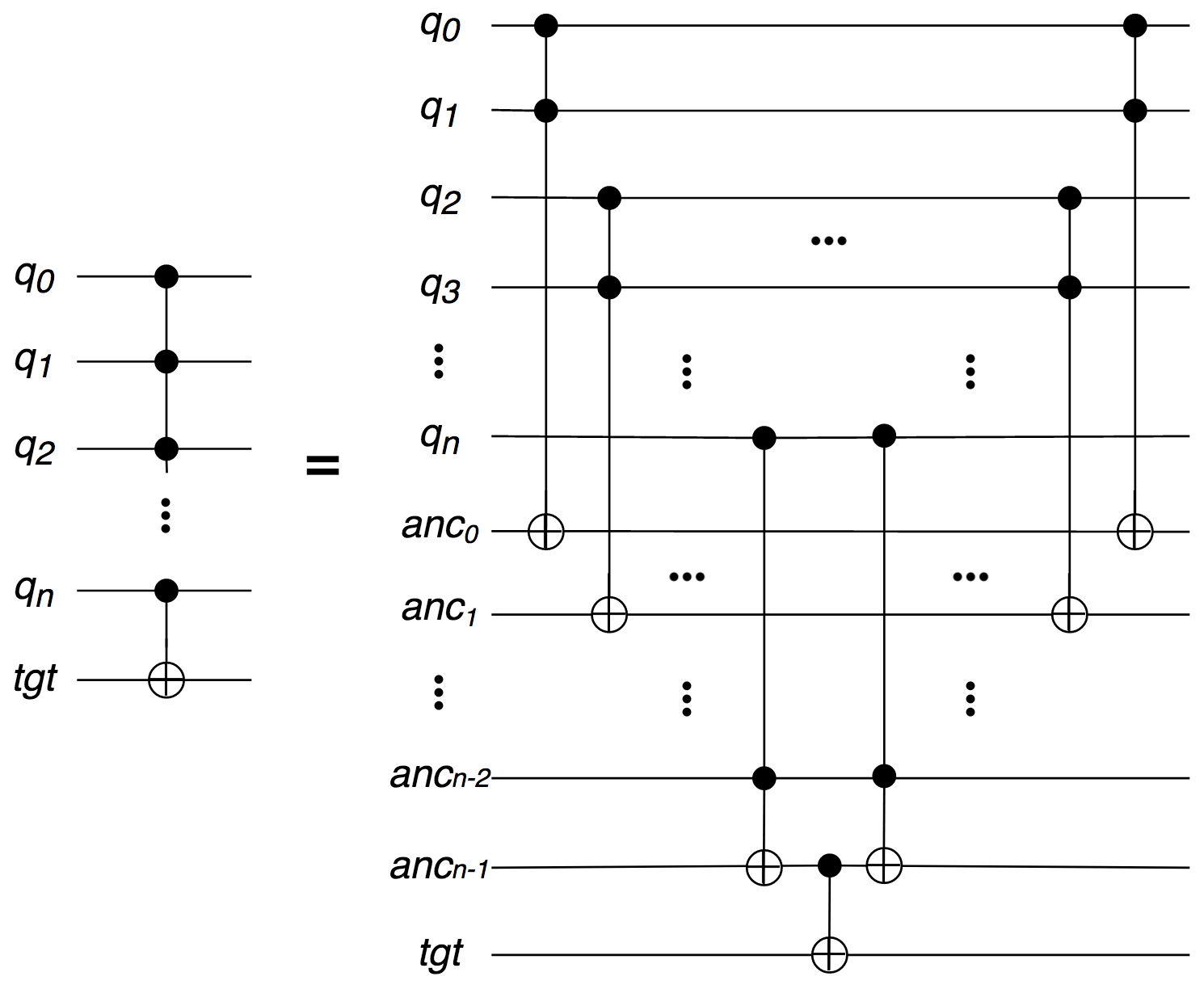}
    \end{center}
    \caption{Generalised Toffoli gate with $n$ control qubits ($q_{0}$ to $q_{n}$), $n-1$ ancilla qubits ($anc_{0}$ to $anc_{n-1}$) and one target qubit ($tgt$).}
    \label{fig:GenTof}
\end{figure}

The dependency of the operations between the qubits within the architectures does not allow for any gates to be run in parallel, restricting the width of the circuit at each time step. Thus, the runtime of the quantum circuit quickly surpasses the coherence time of the qubits, leading to immensely noisy distributions. For very small state spaces (i.e. two qubits) we see distributions closer to the theory with lower level of noise. This does not hold for walks with a three-qubit state space or larger. The execution time of the circuit with relation to the coherence time of the qubits greatly affects the resulting distribution of the quantum walk.

Due to the stochastic nature of the noise, it is very difficult to draw safe conclusions from the comparison of the two approaches on the quantum computer. We can, however, point out that since the rotational circuit is deeper, the cumulative error due to hardware infidelities will be more extensive. Additionally, since the computation also takes longer in the case of the rotational circuit, the active qubits have higher chance to decohere. Thus, we believe it is safe to claim that the rotational circuit will be the noisier of the two.

As for the comparison between the two approaches, we can streamline the discussion using the quantum volume \cite{Moll-2018}, which is an architecture-neutral figure of merit that showcases the performance of a quantum computer when running quantum circuits. It depends on the number of physical qubits in the machine, the number of qubits utilized by the circuit, the depth and width of the circuit, as well as the average effective error rates of the two-qubit gates implemented in the quantum computer.

Following the above discussion, we can calculate the quantum volume that each circuit would require to be run on the same quantum architecture. More specifically, considering as an example one step of a quantum walk with a three qubits state space on IBMQ's Melbourne machine, the generalized inverters approach requires a quantum volume of approximately $20.812$, whereas the rotational approach requires $28.905$ (for more information on the calculation of the quantum volume for the three qubit example see Appendix \ref{ap:qvol}). In a similar way we can calculate the volume for the larger quantum walk circuits. Thus, we can safely conclude that the generalized inverters approach would require the smaller quantum volume of the two implementations due to the much smaller circuit depth and subsequently reduced execution time compared to the rotational approach, as well as lower the cumulative effective error due to hardware infidelities.

It is noteworthy that, to the best of our knowledge, no optimal implementation of quantum walks currently exists. The generalized inverters approach proposed by \cite{EffWalk} constitutes an efficient circuit. Other implementations might exist that could perhaps prove to be more efficient and could be the subject of further comparison.

Finally, due to large presence of noise in the quantum machine, we cannot draw firm conclusions regarding today's quantum hardware. In this paper we have worked out a resource theory for and compared two different implementations of quantum walks. Therefore, considering the limitations and hardware constraints of NISQ machines, our analysis can assist on the realisation of quantum walks on near-term quantum computers in an efficient way.

\section{Acknowledgements}
This work was supported by the Engineering and Physical Sciences Research Council, Centre for Doctoral Training in Cloud Computing for Big Data [grant number EP/L015358/1].

\appendix

\section{Generalised Inverter Gate for Arbitrary Control Qubits}\label{ap:cnx}
In many circumstances, we need to control an inversion with an arbitrary number of $n_{c}>2$ control qubits. A solution can be given by introducing intermediate computations, with their results stored in an ancilla register of size $n_{c}-1$ \cite{NielChu}. A visualization of this solution for a generalised \verb|CNOT| gate with $n_{c}$ control qubits is shown in Figure \ref{fig:GenTof}. This decomposition of the generalised \verb|CNOT| gate can be further simplified to use just regular \verb|CNOT| operations.

\section{Realisation of Rotation Operations with IBM's \texorpdfstring{$U_{3}$}{TEXT} Gate}\label{ap:u3}
Here we show how IBM Qiskit's \cite{Qiskit} $U_{3}$ gate can be used to create a rotation operator of the form $R_{y}(\theta)$. The $U_{3}(\theta,\phi,\lambda)$ gate is a single-qubit gate with three Euler angles. The gate implements the following operator
\begin{equation}
    U_{3}(\theta,\phi,\lambda) = 
    \begin{pmatrix}
        \cos(\theta/2) & -e^{i\lambda}\sin(\theta/2) \\
        e^{i\phi}\sin(\theta/2) & e^{i\lambda+i\phi}\cos(\theta/2) 
    \end{pmatrix} \label{eq:u3}
\end{equation}
For our case we can assign $\theta=\pi/2, \phi=\lambda=0$, thus getting
\begin{equation*}
    R_{y}(\pi/2) = U_{3}(\pi/2,0,0),
\end{equation*}
with matrix representation easily deducted from Equation \eqref{eq:u3}.

\section{Calculating the Quantum Volume}\label{ap:qvol}
The quantum volume, as introduced by Moll et al. \cite{Moll-2018}, is a figure of merit that characterizes the performance of a quantum computer. In the context of this paper, we use the quantum volume in order to attach an additional metric that shows the required resources of each quantum walk approach on IBMQ's $15$-qubit Melbourne machine.

The quantum volume can be defined as \cite{Moll-2018}
\begin{equation}
    V_{Q} = \max_{n<N} \left( \min \left[ n,\frac{1}{n\times\epsilon_{\text{eff}}(n)} \right]^{2} \right) \label{eq:qvol}
\end{equation}
where $n$ is the size of the workspace necessary for the computation, $N$ is the number of qubits within the quantum computer and $\epsilon_{\text{eff}}(n)$ is the average effective two-qubit gate error rates of the qubits that participate in the circuit, following the connectivity of the architecture.

Thus, we can now easily calculate the quantum volume of the three qubit experiment we use as an example. It is important here that, due to automatic optimization of the connectivity before execution in the computer, additional qubits that store intermediate quantum states have to be taken into account when calculating the quantum volume. Subsequently, the size of the workspace for the generalized inverters approach is $n_{c} = 8$ and for the rotational us $n_{r} = 6$. For the qubits that participate in the workspace of each approach, we can compute the average error rates as $\epsilon_{\text{eff}}^{c} = 2.74 \times 10^{-2}$ and $\epsilon_{\text{eff}}^{r} = 3.10 \times 10^{-2}$ for the inverters and rotational approach respectively, as of the day of the experiment.

Substituting the above values in equation \ref{eq:qvol}, we compute the quantum volume for the generalized inverters approach as $V_{Q}^{c} = 20.812$ and for the rotational approach as $V_{Q}^{r} = 28.905$. In a similar fashion we can calculate the quantum volume for a quantum walk of arbitrary size.

\bibliography{main}{}

\end{document}